\pdfoutput=1
\documentclass[11pt]{article}
\usepackage[a4paper,margin=30mm]{geometry}
\usepackage{amsmath,amssymb}
\usepackage{graphicx}
\usepackage[hidelinks]{hyperref}

\newcounter{defthM}
\newenvironment{definition}[1]{\refstepcounter{defthM}\trivlist
\item[\hskip5pt\textbf{#1~\arabic{defthM}.}]\rm\hskip3pt}{\endtrivlist}
\newenvironment{theorem}[1]{\refstepcounter{defthM}\trivlist
\item[\hskip5pt\textbf{#1~\arabic{defthM}.}]\it\hskip3pt}{\endtrivlist}

\title{Entanglement breaking in self-referential quantum feedback:
a record bound and an exactly solvable loop}
\author{Eran Kopel\\
{\footnotesize\it Tel Aviv University, Tel Aviv, Israel}\\
{\footnotesize\it e-mail: erankopel@tauex.tau.ac.il;
ORCID 0000-0003-4657-8636}}

\begin{document}

\date{}
\maketitle

\begin{abstract}
A feedback loop that returns a message into the process it describes
defines a quantum channel on the message. We ask when such a loop
preserves entanglement with a reference, that is, when its channel is not
entanglement breaking. For any qubit loop, the
negativity of the Choi state is at most half the fidelity between the
states that the two values of the message leave in the discarded
registers, in any basis. This channel analogue of the duality between
fringe visibility and which-way information is attained by pure dephasing,
and a perfect record of the message in the discarded registers makes the
loop entanglement breaking. We then solve a three-qubit loop in which the
message steers a future event, a probe reads the event, and a partial swap
feeds the probe back. Every partial feedback is a strict contraction, and
the loop is entanglement breaking exactly when an explicit trigonometric
polynomial in its coupling angles is not positive. A partial record of the
message in the future event can restore entanglement that the feedback
alone destroys. Stability, a coherent fixed point, information about the
future event and entanglement preservation coexist on about three quarters
of the parameter space.
\end{abstract}

\section{Introduction}

Consider a loop in which a message is fed into a process, the process
responds, a probe reads part of the response, and the reading is written
back into the message register before the next round. When the response is
a later event that the message itself influences, the loop is
self-referential: the returned message is both evidence about the event and
an input to it, as for a forecast that changes the outcome it forecasts.
Inside a simulation, the event may lie in the simulated future while the
protocol remains an ordinary, causally ordered circuit. Loops of this kind
are natural objects for the theory of multi-time quantum processes
\cite{Chiribella2008,Chiribella2009,MilzModi2021,Pollock2018}, for
repeated-interaction (collision) models of open systems
\cite{Ciccarello2022}, and for coherent quantum feedback
\cite{Grimsmo2015,Lloyd2000,WisemanMilburn2009}.

One round of such a loop induces a completely positive, trace-preserving
(CPTP) map $\Phi$ on the message register, and the long-run behaviour of the
loop is set by the fixed points of $\Phi$. Four properties of the loop are
of interest and are easily conflated: stability (a unique, attracting fixed
point), coherence of the stationary message, information about the future
event held by the probe, and the ability of the loop to preserve
entanglement between the message and a reference system, that is, the
property that $\Phi$ is not entanglement breaking (EB)
\cite{HSR2003,Ruskai2003}. A coherent stationary message does not imply the
last property, since an EB channel can prepare any state.

This paper makes two contributions. The first is a general bound
(Theorem~\ref{thm:record}). Each register that a round discards acts as a
which-way detector for the message: the better the discarded registers can
tell the two values of the message apart, in any basis, the less
entanglement the loop can preserve. Quantitatively, the Choi negativity of a
qubit loop is at most half the fidelity between the two corresponding
states of the discarded registers. Duality relations between the coherence
of a state and the path information stored in a detector are well studied
\cite{Bagan2016,Englert1996}, as are information-disturbance trade-offs
\cite{BHH2008,KSW2008}. The bound carries the duality over from the
coherence of one state to the entanglement that a channel preserves, by
means of Uhlmann's theorem, and it is attained by pure dephasing. A loop
whose discarded registers keep a perfect record of the message is EB.

The second contribution is an exact solution of a three-qubit loop, in
which the message steers a future event, a probe reads the event, and a
partial swap feeds the probe back (Fig.~\ref{fig:round}). We give its Bloch
representation in closed form (Theorem~\ref{thm:closed}), show that every
partial feedback is a strict contraction (Corollary~\ref{cor:contraction}),
and find the exact boundary between EB and non-EB loops
(Theorem~\ref{thm:boundary}). The boundary recovers two ways in which the
loop breaks entanglement: a perfect record of the message in the future
event, and a full swap that discards the old message. It shows that the final
controller rotation never matters for entanglement
breaking, and it has a feature that the
bound alone does not predict: a partial record in the future event can make
a loop non-EB that is EB without one (Corollary~\ref{cor:reentrant}). The
four properties coexist on about three quarters of the parameter space
(Proposition~\ref{prop:generic}) and on an explicit square checked in ball
arithmetic (Proposition~\ref{prop:square}).

An earlier version of this work \cite{Kopel2026a} established the
coexistence of the four properties on a square of half-width $0.0013\pi$ in
two of the angles, using interval perturbation bounds and a machine-verified
certificate. The closed forms below supersede that analysis: they reproduce
every reference value of the earlier version, and the square of
Proposition~\ref{prop:square} is 77 times wider. Iterates of the same loop,
and the number of rounds after which they become EB, are studied in
\cite{Kopel2026b,Kopel2026c}.

Nothing here involves signalling into the physical past. The future event
is later only on the internal clock of the simulated process; externally
the protocol is a fixed circuit. The loop shares with Deutsch's model of
closed timelike curves
\cite{Deutsch1991} only its fixed-point equation, and it uses no
postselection, unlike the model of closed timelike curves based on
postselected teleportation \cite{Lloyd2011}. Here the channel is a fixed
linear map, so a unique attracting fixed
point is the generic case \cite{Wolf2012}.

\section{Setting}

\subsection{Loops and their channels}

A round acts on the message register $M$ and on auxiliary registers that
start in fixed pure states; in our model these are a future register $F$
and a probe $L$, each prepared in $|0\rangle$. A round is a unitary $U$ on
$M\otimes F\otimes L$, and the message channel is
\begin{equation}
\Phi(\rho)=\mathrm{Tr}_{FL}\big[\,U(\rho\otimes|00\rangle\langle 00|_{FL})
U^\dagger\big].
\label{eq:channel}
\end{equation}
Equivalently, $\Phi$ has the Stinespring isometry
$V|\psi\rangle=U(|\psi\rangle\otimes|00\rangle_{FL})$, and the
complementary channel
\begin{equation}
\Phi^{c}(\rho)=\mathrm{Tr}_{M}\big[V\rho V^\dagger\big]
\label{eq:complementary}
\end{equation}
describes what the discarded registers receive \cite{DevetakShor2005}. The
general results of Section~3 concern $\Phi$ and $\Phi^c$ only, so they
apply to any loop of this structure, for instance one in which the process
tensor of a larger simulation is contracted with a probe instrument and a
controller \cite{MilzModi2021,Pollock2018}. The auxiliaries are fresh in
every round, so the loop is memoryless and $n$ rounds act as $\Phi^n$.

\subsection{Four properties}

We write a qubit state as $\rho=(1+\mathbf{v}\cdot\boldsymbol{\sigma})/2$
with $\boldsymbol{\sigma}=(X,Y,Z)$. The channel acts on Bloch vectors as
$\mathbf{v}\mapsto A\mathbf{v}+\mathbf{c}$, with
$A_{ij}=\frac12\mathrm{Tr}[\sigma_i\Phi(\sigma_j)]$ and
$c_i=\mathrm{Tr}[\sigma_i\Phi(1/2)]$. The four properties are as follows.

{\it Stability.} If $\|A\|_2<1$ (spectral norm), the map is a strict
contraction of the Bloch ball and has a unique fixed point $\mathbf{v}^*$
that attracts every initial message. We call $1-\|A\|_2$ the contraction
margin.

{\it Coherence.} We use the $\ell_1$ coherence of the stationary message in
the computational basis, the basis in which the forecast is read
\cite{Baumgratz2014,Streltsov2017}:
$\mathcal{C}(\rho^*)=(x^{*2}+y^{*2})^{1/2}$.

{\it Information.} The quantum mutual information $I(F\!:\!L)$ of the
future register and the probe immediately before the feedback, with the
stationary message as input, measures how much the probe has learned about
the future event.

{\it Entanglement preservation.} Let
$J(\Phi)=(\mathrm{id}\otimes\Phi)(|\Omega\rangle\langle\Omega|)$, with
$|\Omega\rangle=(|00\rangle+|11\rangle)/\sqrt2$, be the normalized Choi
state, and $J^\Gamma$ its partial transpose on the reference. For qubit
channels, $\Phi$ is EB if and only if $J(\Phi)$ is separable \cite{HSR2003},
if and only if $J^\Gamma\ge0$ \cite{Horodecki1996,Peres1996}. We use the
signed margin
\begin{equation}
\mu(\Phi)=-\lambda_{\min}\big(J(\Phi)^\Gamma\big),
\label{eq:mu}
\end{equation}
so that $\Phi$ is not EB if and only if $\mu(\Phi)>0$. Since the partial
transpose of a two-qubit state has at most one negative eigenvalue
\cite{Sanpera1998}, $\max\{\mu,0\}$ equals the negativity
$(\|J^\Gamma\|_1-1)/2$.

\section{The record bound}

Let $F(\rho,\sigma)=\|\sqrt\rho\sqrt\sigma\|_1$ denote the fidelity and
$T(\rho,\sigma)=\frac12\|\rho-\sigma\|_1$ the trace distance. The
following lemma is Uhlmann's theorem, stated for a channel and its
complement.

\begin{theorem}{Lemma}\label{lem:transfer}
Let $\Phi$ be a channel with complementary channel $\Phi^c$, and let
$|e_0\rangle$, $|e_1\rangle$ be orthonormal. Then
\begin{equation}
\big\|\Phi(|e_1\rangle\langle e_0|)\big\|_1
=F\big(\Phi^c(|e_0\rangle\langle e_0|),\,\Phi^c(|e_1\rangle\langle e_1|)
\big).
\label{eq:transfer}
\end{equation}
\end{theorem}

{\it Proof.} The vectors $|v_j\rangle=V|e_j\rangle$ purify
$\rho_E^j=\Phi^c(|e_j\rangle\langle e_j|)$, with the message register as
purifying system, and
$\Phi(|e_1\rangle\langle e_0|)=\mathrm{Tr}_E|v_1\rangle\langle v_0|$. For
two purifications on a common space, Uhlmann's theorem in the form
$F(\rho_E^0,\rho_E^1)=\|\mathrm{Tr}_E|v_1\rangle\langle v_0|\|_1$
\cite{Uhlmann1976,Watrous2018} gives the claim. \hfill\rule{5pt}{5pt}

\begin{theorem}{Theorem}\label{thm:record}
Let $\Phi$ be a qubit channel, let $\{|e_0\rangle,|e_1\rangle\}$ be any
orthonormal basis, and let $\rho_E^j=\Phi^c(|e_j\rangle\langle e_j|)$. Then
\begin{equation}
\mu(\Phi)\;\le\;\tfrac12\big\|\Phi(|e_1\rangle\langle e_0|)\big\|_\infty
\;\le\;\tfrac12 F\big(\rho_E^0,\rho_E^1\big).
\label{eq:record}
\end{equation}
\end{theorem}

{\it Proof.} Let $W$ be the unitary with $W|j\rangle=|e_j\rangle$ and
$\Phi_W=\Phi\circ\mathrm{Ad}_W$. Since $(1\otimes W)|\Omega\rangle=
(W^{T}\otimes1)|\Omega\rangle$, we have
$J(\Phi_W)=(W^T\otimes1)J(\Phi)(\overline{W}\otimes1)$ and hence
$J(\Phi_W)^\Gamma=(W^\dagger\otimes1)J(\Phi)^\Gamma(W\otimes1)$, so
$\mu(\Phi_W)=\mu(\Phi)$. Moreover $\Phi_W(|1\rangle\langle0|)=
\Phi(|e_1\rangle\langle e_0|)$ and $\Phi_W^c(|j\rangle\langle j|)=\rho_E^j$,
so it suffices to take $|e_j\rangle=|j\rangle$. In the computational basis
of the reference,
\begin{equation}
J^\Gamma=\tfrac12\sum_{m,n}|n\rangle\langle m|\otimes\Phi(|m\rangle\langle n|)
=\tfrac12(D+O),
\end{equation}
where $D=\sum_m|m\rangle\langle m|\otimes\Phi(|m\rangle\langle m|)\ge0$ and
$O=|0\rangle\langle1|\otimes\Phi(|1\rangle\langle0|)+
|1\rangle\langle0|\otimes\Phi(|0\rangle\langle1|)$. The Hermitian matrix
$O$ has eigenvalues $\pm s_k$, where $s_k$ are the singular values of
$\Phi(|1\rangle\langle 0|)$. By Weyl's inequality,
$\lambda_{\min}(J^\Gamma)\ge\frac12[\lambda_{\min}(D)+\lambda_{\min}(O)]
\ge-\frac12\|\Phi(|1\rangle\langle0|)\|_\infty$. The second inequality in
(\ref{eq:record}) follows from $\|\cdot\|_\infty\le\|\cdot\|_1$ and
Lemma~\ref{lem:transfer}. \hfill\rule{5pt}{5pt}

\begin{theorem}{Corollary}\label{cor:duality}
In the setting of Theorem~\ref{thm:record}, let $\mu_+=\max\{\mu,0\}$
and $T_E=T(\rho_E^0,\rho_E^1)$. Then:
(i) $(2\mu_+)^2+T_E^2\le1$;
(ii) the bounds remain valid when $\rho_E^j$ are replaced by their
reductions to any subset of the discarded registers;
(iii) if the discarded registers, or any subset of them, distinguish
$|e_0\rangle$ from $|e_1\rangle$ perfectly, then $\Phi$ is EB.
\end{theorem}

{\it Proof.} (i) follows from $F\le(1-T^2)^{1/2}$ \cite{FvdG1999}. (ii)
The fidelity does not decrease under partial traces. (iii) If $F=0$,
Lemma~\ref{lem:transfer} gives $\Phi(|e_1\rangle\langle e_0|)=0$, so
$\Phi=\Phi\circ\Delta_e$, with $\Delta_e$ the complete dephasing in the
basis $\{|e_j\rangle\}$. This is a measure-and-prepare channel, hence EB
\cite{HSR2003}. \hfill\rule{5pt}{5pt}

\begin{definition}{Remark}\label{rem:tight}
The bounds are tight. For pure dephasing, $V|j\rangle=|j\rangle\otimes
|\chi_j\rangle$, one finds $\mu=\frac12|\langle\chi_0|\chi_1\rangle|=
\frac12F(\rho_E^0,\rho_E^1)$, and since the environment states are pure,
$(2\mu)^2+T_E^2=1$. For a $d$-level message the same argument, with Schur's
test applied to the off-diagonal blocks of $J^\Gamma$, gives
$\lambda_{\min}(J^\Gamma)\ge-\frac1d\max_n\sum_{m\ne n}F(\rho_E^m,
\rho_E^n)$, and part (iii) of Corollary~\ref{cor:duality} holds when the
discarded
registers distinguish all $d$ basis states perfectly.
\end{definition}

Theorem~\ref{thm:record} turns any which-way record into an entanglement
budget, in every basis at once. In a self-referential loop the future event
depends on the message and therefore records it. By
Corollary~\ref{cor:duality}(ii), applied to the future register alone, a
responsive future costs entanglement: if the two values of the message
steer the future register into states at trace distance $T_F$, then
$\mu\le\frac12(1-T_F^2)^{1/2}$. Feedback adds a second record, since
whatever part of the old message the feedback moves into the probe is
discarded with it. The model of the next section shows how the two records
combine.

\section{An exactly solvable loop}

\subsection{The loop}

\begin{figure}[t]
\centering
\includegraphics[width=0.92\textwidth]{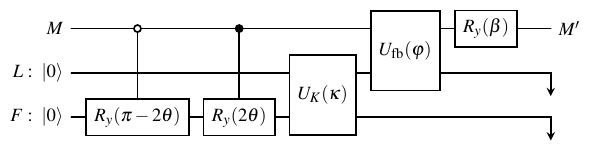}
\caption{One round of the loop. The message $M$ steers the future register
$F$ (open control: $M=|0\rangle$; filled control: $M=|1\rangle$), the probe
$L$ reads $F$, the partial swap $U_{\mathrm{fb}}$ exchanges probe and
message, and the controller rotates the message. $F$ and $L$ are discarded;
fresh copies are used in the next round.}
\label{fig:round}
\end{figure}

One round (Fig.~\ref{fig:round}) is
$U=R_y^{M}(\beta)\,U_{\mathrm{fb}}(\varphi)\,U_K(\kappa)\,U_W(\theta)$, with
$R_y(\alpha)=e^{-i\alpha Y/2}$ and
\begin{align}
U_W(\theta)&=|0\rangle\langle0|_M\otimes R_y(\pi-2\theta)_F
+|1\rangle\langle1|_M\otimes R_y(2\theta)_F,
\label{eq:UW}\\
U_K(\kappa)&=\exp\big[-i\tfrac{\kappa}{2}\,Z_F\otimes Y_L\big],
\label{eq:UK}\\
U_{\mathrm{fb}}(\varphi)&=\cos\varphi\,1-i\sin\varphi\,
\mathrm{SWAP}_{ML}.
\label{eq:Ufb}
\end{align}
The world unitary $U_W$ steers the future register into
$|\psi_0\rangle=\sin\theta|0\rangle+\cos\theta|1\rangle$ or
$|\psi_1\rangle=\cos\theta|0\rangle+\sin\theta|1\rangle$, according to the
message. At $\theta=0$ the event contradicts the message perfectly (an
anti-copy), at $\theta=\pi/4$ it does not depend on the message, and
$\langle\psi_0|\psi_1\rangle=\sin2\theta$. The coupling $U_K$ reads the
event in its computational basis and leaves the probe in
$|\lambda_f\rangle=R_y((-1)^f\kappa)|0\rangle$ when the event is $f$, with
$\langle\lambda_0|\lambda_1\rangle=\cos\kappa$. The partial swap feeds the
probe into the message register and, by the same token, moves the old
message into the probe, which is then discarded. The controller
$R_y(\beta)$ rotates the returned message. We take
$\theta,\varphi\in[0,\pi/2]$, $\kappa\in[0,\pi]$ and $\beta\in[0,2\pi)$.

A classical version of the loop shows what can go wrong. If the message
is dephased and the probe reads the event perfectly, the loop is an
anti-predictor: at perfect anti-copy a definite message is always
contradicted and the forecast flips every round. The only stationary
distribution is the uniform one, and it attracts nothing
\cite{Kopel2026a}. The quantum loop has a unique attracting fixed
point for every partial feedback (Corollary~\ref{cor:contraction}).

\subsection{Closed form}

\begin{theorem}{Theorem}\label{thm:closed}
For all $\theta,\kappa,\varphi,\beta$, the Bloch data of the loop are
$A=R_\beta\tilde A$ and $\mathbf{c}=R_\beta\tilde{\mathbf c}$, where
$R_\beta$ rotates Bloch vectors by $\beta$ about the $y$ axis and
\begin{equation}
\tilde A=\begin{pmatrix}
s\cos^2\varphi & -s\nu & -w\\
s\nu & s\cos^2\varphi & 0\\
0 & 0 & \cos^2\varphi\end{pmatrix},\qquad
\tilde{\mathbf c}=\begin{pmatrix}0\\
\sin\kappa\sin\varphi\cos\varphi\cos2\theta\\
\sin^2\varphi\cos\kappa\end{pmatrix},
\label{eq:closed}
\end{equation}
with $s=\sin2\theta$, $\nu=\sin\varphi\cos\varphi\cos\kappa$ and
$w=\sin\kappa\sin^2\varphi\cos2\theta$.
\end{theorem}

The proof is in Appendix~A. The upper-left block of $\tilde A$ is a
multiple $b\ge0$ of a rotation about the $z$ axis, with
\begin{equation}
b^2=\sin^22\theta\cos^2\varphi\,(\cos^2\varphi+\sin^2\varphi
\cos^2\kappa).
\label{eq:b}
\end{equation}
Without the controller, the homogeneous part of the transverse dynamics
therefore shrinks by $b$ and turns at each round, while the $z$ component,
which also feeds the $x$ component through the entry $-w$, relaxes
geometrically with ratio $\cos^2\varphi$.

\begin{theorem}{Corollary}\label{cor:contraction}
$\|A\|_2$ does not depend on $\beta$, and
\begin{equation}
\|A\|_2^2=\tfrac12\Big[\tau+\big(\tau^2-4b^2\cos^4\varphi\big)^{1/2}\Big],
\qquad \tau=b^2+w^2+\cos^4\varphi,
\label{eq:norm}
\end{equation}
with $b$ and $w$ as in Eqs.~(\ref{eq:b}) and (\ref{eq:closed}). For
$0<\varphi<\pi/2$ and all $\theta,\kappa,\beta$, $\|A\|_2<1$: the loop is a
strict contraction, and its fixed point $\mathbf v^*=(1-A)^{-1}\mathbf c$
attracts every initial message.
\end{theorem}

{\it Proof.} $R_\beta$ is orthogonal, so $\|A\|_2=\|\tilde A\|_2$. Write
the upper-left block of $\tilde A$ as $B$, so that $B^TB=BB^T=b^2\,1$, and
recall that the third column is $(-w,0,\cos^2\varphi)^T$. The
plane spanned by $(B^T\mathbf e_1,0)$ and the third axis is invariant under
$\tilde A^T\tilde A$, which acts on it as
$\mathcal M=\left(\begin{smallmatrix}b^2 & -bw\\ -bw & w^2+\cos^4\varphi
\end{smallmatrix}\right)$, and on the orthogonal direction as $b^2\le
\lambda_{\max}(\mathcal M)$. This gives (\ref{eq:norm}). For
$0<\varphi<\pi/2$, $b^2\le\sin^22\theta\cos^2\varphi<1$ and
\begin{align*}
\det(1-\mathcal M)&=(1-b^2)(1-\cos^4\varphi)-w^2\\
&\ge(\cos^22\theta+\sin^22\theta\sin^2\varphi)\sin^2\varphi
(1+\cos^2\varphi)-\sin^4\varphi\cos^22\theta\\
&=\sin^2\varphi\,\big[2\cos^22\theta\cos^2\varphi
+\sin^22\theta\sin^2\varphi(1+\cos^2\varphi)\big]>0,
\end{align*}
so both eigenvalues of $\mathcal M$ are smaller than one. The fixed point
then follows from Banach's theorem. \hfill\rule{5pt}{5pt}

\begin{theorem}{Proposition}\label{prop:info}
Before the feedback, the future register and the probe are in the state
$\rho_{FL}=\sum_m p_m|\Psi_m\rangle\langle\Psi_m|$, with
$|\Psi_m\rangle=U_K(|\psi_m\rangle\otimes|0\rangle)$ and $p_0=(1+z)/2$,
where $z$ is the $Z$ component of the message. With $h(r)$ the entropy in
bits of a qubit state with Bloch length $r$,
\begin{equation}
I(F\!:\!L)=h(r_F)+h(r_L)-h(r_{FL}),
\label{eq:info}
\end{equation}
where $r_F^2=z^2\cos^22\theta+\sin^22\theta\cos^2\kappa$,
$r_L^2=z^2\cos^22\theta\sin^2\kappa+\cos^2\kappa$ and
$r_{FL}^2=\sin^22\theta+z^2\cos^22\theta$. For $0<\theta<\pi/2$ and
$0<\kappa<\pi$, $I(F\!:\!L)>0$ for every message.
\end{theorem}

{\it Proof.} Since $U_K$ multiplies the off-diagonal entries of $F$ by
$\langle\lambda_1|\lambda_0\rangle=\cos\kappa$, the Bloch vector of
$\rho_F$ is $(\sin2\theta\cos\kappa,\,0,\,-z\cos2\theta)$. The probe
is in $\sum_fw_f|\lambda_f\rangle\langle\lambda_f|$, with weight
\begin{equation*}
w_0=p_0\sin^2\theta+p_1\cos^2\theta=\tfrac12(1-z\cos2\theta)
\end{equation*}
and Bloch vector $(-z\cos2\theta\sin\kappa,\,0,\,\cos\kappa)$. The two
vectors $|\Psi_m\rangle$ have overlap $\sin2\theta$, so $\rho_{FL}$ has
eigenvalues $\frac12(1\pm r_{FL})$ and two zeros. If $I(F\!:\!L)=0$,
then $\rho_{FL}=\rho_F\otimes\rho_L$, and projecting $F$ onto
$|f\rangle$ would give
\begin{equation*}
\mathrm{Tr}_F\big[(|f\rangle\langle f|\otimes1)\rho_{FL}\big]
=w_f|\lambda_f\rangle\langle\lambda_f|\;\propto\;\rho_L,
\qquad f=0,1.
\end{equation*}
For $0<\theta<\pi/2$ both weights $w_f$ are positive, so $\rho_L$ would be
proportional to two different pure states. \hfill\rule{5pt}{5pt}

\begin{definition}{Example}\label{ex:coherent}
At $\theta=0$ and $\varphi=\pi/2$, Eq.~(\ref{eq:closed}) gives the fixed
point
\begin{equation}
z^*=\frac{\cos\beta\cos\kappa}{1-\sin\beta\sin\kappa},\qquad
x^*=\frac{\cos\kappa(\sin\beta-\sin\kappa)}{1-\sin\beta\sin\kappa},\qquad
y^*=0,
\end{equation}
which attracts every message for $0\le\kappa,\beta<\pi/2$ (the spectral
radius of $A$ is $\sin\beta\sin\kappa$). It is coherent unless
$\beta=\kappa$. Yet the loop is EB: both the future register (perfect
anti-copy) and the probe (full swap) hold a perfect record of the message,
and Corollary~\ref{cor:duality}(iii) applies. A coherent stationary message
is therefore no evidence of entanglement preservation.
\end{definition}

\section{The exact entanglement-breaking boundary}

\begin{theorem}{Theorem}\label{thm:boundary}
Let $C=\cos^22\theta$, $S=\sin^2\varphi$, $u=\sin^2(\kappa/2)$ and
$q=u(1-u)=\frac14\sin^2\kappa$. For all $\theta,\kappa,\varphi,\beta$,
$16\det J(\Phi)^\Gamma=-Q(\theta,\kappa,\varphi)$, where
\begin{equation}
\begin{split}
Q={}&(1-C)(1-S)(1-4qS)(1-uS)\big(1-(1-u)S\big)\\
&-qS^2\big(1-uS-(1-u)C\big)\big(1-(1-u)S-uC\big).
\end{split}
\label{eq:Q}
\end{equation}
The loop is not EB if and only if $Q>0$, and then
$Q/2\le\mu(\Phi)\le\frac12$.
\end{theorem}

{\it Proof.} Appendix~A shows that $J^\Gamma$ is, up to a factor and a
unitary on the output, a tridiagonal matrix, and computes its determinant.
A two-qubit state is entangled if and only if the determinant of its
partial transpose is negative \cite{Augusiak2008,Sanpera1998}; this gives
the first claim. If $Q>0$, $J^\Gamma$ has one negative eigenvalue, $-\mu$,
and three positive ones \cite{Sanpera1998,Verstraete2001}, which sum to
$1+\mu$. For every unit vector $|\xi\rangle$,
$\langle\xi|J^\Gamma|\xi\rangle=\mathrm{Tr}[J(\Phi)\,(|\xi\rangle
\langle\xi|)^\Gamma]\ge-\frac12$, because $J(\Phi)$ is a state and the
partial transpose of a pure two-qubit state with Schmidt coefficients
$\lambda_1,\lambda_2$ ($\lambda_1^2+\lambda_2^2=1$) has smallest
eigenvalue
$-\lambda_1\lambda_2\ge-\frac12$. Hence $\mu\le\frac12$, and by the
inequality of arithmetic and geometric means the three positive
eigenvalues have product at most $((1+\mu)/3)^3\le\frac18$. So
$Q/16=|\det J^\Gamma|\le\mu/8$. \hfill\rule{5pt}{5pt}

\begin{figure}[t]
\centering
\includegraphics[width=\textwidth]{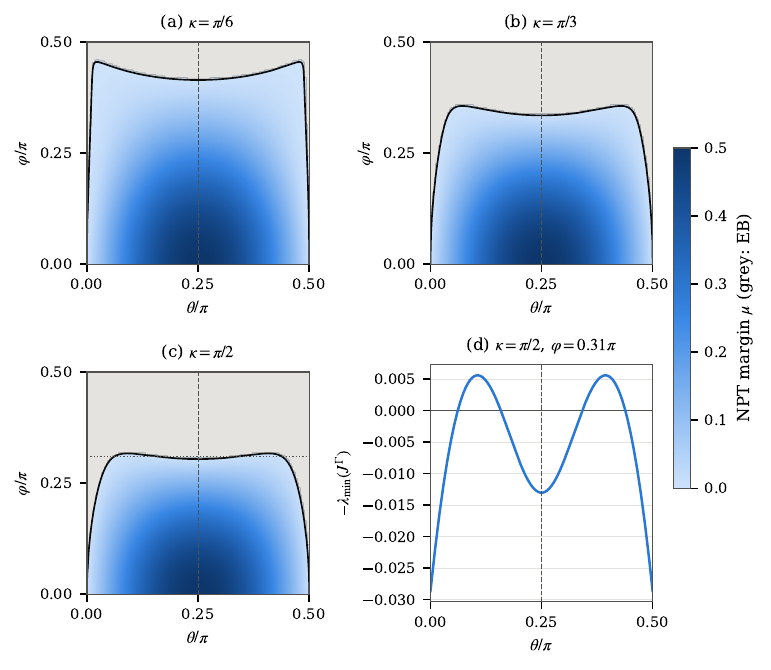}
\caption{(a)--(c) Entanglement-breaking phase diagram in the
$(\theta,\varphi)$ plane for three probe couplings. Colour: NPT margin
$\mu$ where the loop is not EB; grey: EB. The black line is the exact
boundary $Q=0$ of Theorem~\ref{thm:boundary}; the dashed line $\theta=\pi/4$
marks a future that does not depend on the message. The dotted line in (c)
is $\varphi=0.31\pi$. (d) Signed margin $-\lambda_{\min}(J^\Gamma)$ along
that line: the loop is EB at $\theta=0$ and at $\theta=\pi/4$ but not EB in
between (Corollary~\ref{cor:reentrant}).}
\label{fig:phase}
\end{figure}

\begin{theorem}{Corollary}\label{cor:mechanisms}
(i) Perfect record in the future register, $\theta\in\{0,\pi/2\}$:
$Q=-q^2S^2(1-S)^2\le0$, and the loop is EB for all $\kappa,\varphi,\beta$.
(ii) Full feedback, $\varphi=\pi/2$: $Q=-q^2(1-C)^2\le0$, and the loop is
EB. (iii) No feedback, $\varphi=0$: $Q=1-C$; the channel is the pure
dephasing of Remark~\ref{rem:tight} by the future record, not EB unless
$\theta\in\{0,\pi/2\}$. (iv) No probe coupling, $\kappa\in\{0,\pi\}$:
$Q=(1-C)(1-S)^2$. (v) $Q$ does not depend on $\beta$, and it is invariant
under $\theta\mapsto\pi/2-\theta$ and $\kappa\mapsto\pi-\kappa$.
\end{theorem}

{\it Proof.} Substitution in (\ref{eq:Q}); for (v), note that $C$, $S$ and
$q$ are invariant and that $u\mapsto1-u$ exchanges the last two factors of
(\ref{eq:Q}). \hfill\rule{5pt}{5pt}

Cases (i) and (ii) are the two records of Section~3. Indeed, by
Appendix~A, $\Phi(|1\rangle\langle0|)$ is $\gamma\,|1\rangle\langle0|$ up to
the controller, with
\begin{equation}
\gamma=\sin2\theta\cos\varphi\,(\cos\varphi+i\sin\varphi\cos\kappa),
\label{eq:gamma}
\end{equation}
so Theorem~\ref{thm:record} gives
\begin{equation}
\mu\;\le\;\tfrac12|\gamma|=\tfrac12\sin2\theta\cos\varphi
\big(1-\sin^2\varphi\sin^2\kappa\big)^{1/2}.
\label{eq:gammabound}
\end{equation}
The fidelity of the records factorizes: $\sin2\theta$ is the overlap of the
two future events, and the remaining factor comes from what the feedback
leaves in the probe. It vanishes in cases (i) and (ii), and at $\varphi=0$
the bound is an equality. Figure~\ref{fig:duality} compares $\mu$ with the
bound for the loop and for random qubit channels.

\begin{figure}[t]
\centering
\includegraphics[width=\textwidth]{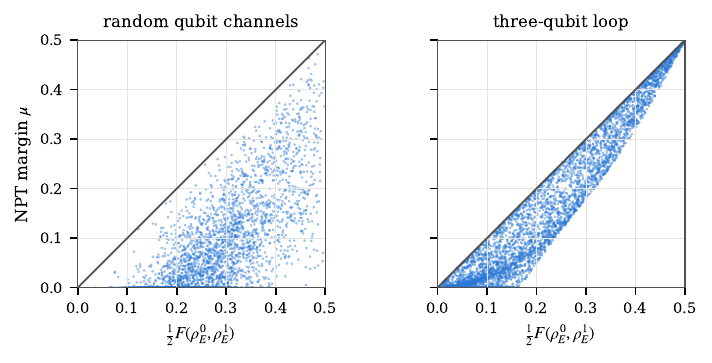}
\caption{The record bound of Theorem~\ref{thm:record} in the computational
basis. Left: 3000 random qubit channels with environments of dimension 2
to 6. Right: 6000 random parameter points of the loop, where $F=|\gamma|$ by
Lemma~\ref{lem:transfer} and Eq.~(\ref{eq:A3}). Every point lies
on or below the diagonal. In the right panel, points on the diagonal are
loops with weak feedback ($\varphi\approx0$), close to pure dephasing.}
\label{fig:duality}
\end{figure}

The bound is not the whole story, because the two records interact.

\begin{theorem}{Corollary}\label{cor:reentrant}
(i) At $\theta=\pi/4$, where the future does not depend on the message, the
loop is not EB if and only if $\sin^2\varphi<S^*(\kappa)$, with
\begin{equation}
S^*(\kappa)=\frac{1+\sin^2\kappa-(1-\sin^2\kappa+\sin^4\kappa)^{1/2}}
{\tfrac32\sin^2\kappa}
\label{eq:Sstar}
\end{equation}
for $\sin\kappa\neq0$; for instance $S^*(\pi/2)=2/3$, that is
$\varphi^*=0.3041\pi$. (ii) For fixed $\kappa$ and $\varphi$, the values of
$\cos^22\theta$ for which the loop is not EB form an interval. In
$\theta$, the non-EB set is symmetric about $\pi/4$; it is empty, one
interval containing $\pi/4$, or two mirror intervals. (iii) The last case
occurs:
at $\kappa=\pi/2$ and $\varphi=0.31\pi$ the loop is EB at $\theta=0$ and at
$\theta=\pi/4$, but not EB for $0.0615\pi<\theta<0.1584\pi$ and on the
mirror interval, with $\mu$ up to $0.0056$ at $\theta\approx0.106\pi$.
\end{theorem}

{\it Proof.} (i) At $C=0$, $Q=(1-uS)(1-(1-u)S)\,[1-(1+4q)S+3qS^2]$. The
first two factors are positive for $S<1$, and the bracket is positive
exactly below its smaller root, which is (\ref{eq:Sstar}). (ii) As a
function of $C$, $Q$ is quadratic with leading coefficient $-q^2S^2\le0$,
hence concave, and $Q\le0$ at $C=1$ by Corollary~\ref{cor:mechanisms}(i).
So $\{C\in[0,1]:Q>0\}$ is an interval, and $C=\cos^22\theta$ is symmetric
about $\theta=\pi/4$ and monotone on each side. (iii) Direct evaluation of
(\ref{eq:Q}) and of $\mu$; see Fig.~\ref{fig:phase}(d).
\hfill\rule{5pt}{5pt}

Part (iii) shows that a record of the message in the future event is not
simply a loss. The probe reads that record and the swap returns part of it
to the message register; close to the threshold of part (i), this can
outweigh what the record takes away. The bound (\ref{eq:gammabound}) is
consistent with this, since it limits $\mu$ but does not order the loops.

\section{Coexistence of the four properties}

\begin{theorem}{Proposition}\label{prop:generic}
On the open set $\Omega$ of parameters with $0<\theta<\pi/2$,
$0<\kappa<\pi$, $0<\varphi<\pi/2$, any $\beta$, and $Q(\theta,\kappa,
\varphi)>0$, the loop is a strict contraction, its probe holds positive
information about the future event, and it is not EB. Its fixed point is
coherent outside a closed subset of $\Omega$ of measure zero. For every
$\beta$, the condition $Q>0$ holds on a fraction $0.75945$ of the box
$(0,\pi/2)\times(0,\pi)\times(0,\pi/2)$ of $(\theta,\kappa,\varphi)$.
\end{theorem}

{\it Proof.} The first three claims are Corollary~\ref{cor:contraction},
Proposition~\ref{prop:info} and Theorem~\ref{thm:boundary}. The component
$x^*$ of $\mathbf v^*=(1-A)^{-1}\mathbf c$ is a real-analytic function on
the connected set where $0<\theta<\pi/2$, $0<\kappa<\pi$,
$0<\varphi<\pi/2$, since $1-A$ is invertible there. It is not identically
zero (it equals $0.0928$ at the point (\ref{eq:ref}) below), so its zero
set has measure zero \cite{Mityagin2020}, and $\mathcal C(\rho^*)=0$ only
where $x^*=0$. For the fraction, the measure in $\theta$ is computed exactly,
since $Q$
is quadratic in $\cos^22\theta$, and then averaged by the midpoint rule on
a $2000\times2000$ grid in $(\kappa,\varphi)$; finer grids change the
result by less than $10^{-5}$. \hfill\rule{5pt}{5pt}

The earlier version of this work used the reference point
\begin{equation}
(\theta_0,\kappa_0,\varphi_0,\beta_0)/\pi=(0.16345853,\;0.43230980,\;
0.20061939,\;0.23903823),
\label{eq:ref}
\end{equation}
found by a seeded random search. The closed forms give, at this point,
$\|A\|_2=0.71331$, $\|(1-A)^{-1}\|_2=2.38726$,
$\mathbf v^*=(0.09282,0.56347,0.12724)$ with $\|\mathbf v^*\|=0.58507$,
$\mathcal C(\rho^*)=0.57107$, $I(F\!:\!L)=1.56868$ bits, $\mu=0.20528$
and $Q=0.20381$. The six quantities enclosed by intervals in
\cite{Kopel2026a}, namely $\|A\|_2$, $\|(1-A)^{-1}\|_2$,
$\|\mathbf v^*\|$, $\mathcal C$, $I$ and $\mu$, all lie in those
intervals.

\begin{theorem}{Proposition}\label{prop:square}
Fix $\kappa=\kappa_0$ and $\beta=\beta_0$. For every $(\theta,\varphi)$ in
the square $|\theta/\pi-0.16345853|\le0.1$,
$|\varphi/\pi-0.20061939|\le0.1$,
\begin{equation}
1-\|A\|_2\ge0.08,\quad\mathcal C(\rho^*)\ge0.025,\quad
I(F\!:\!L)\ge1.0\ \mathrm{bit},\quad Q\ge0.002,
\end{equation}
so the loop is not EB, with $\mu\ge10^{-3}$ by
Theorem~\ref{thm:boundary}.
\end{theorem}

{\sloppy{\it Proof.} Computer-assisted. The script
\texttt{certify\_square.py} evaluates the closed forms
(\ref{eq:closed})--(\ref{eq:info}) and (\ref{eq:Q}), and solves for
$\mathbf v^*$, in ball arithmetic (Arb through python-flint, 128-bit
precision). It covers the square with 19\,354 sub-boxes; on each, the
computed enclosures prove all four inequalities. The proof is conditional
on the correctness of the arithmetic library.
\hfill\rule{5pt}{5pt}\par}

\begin{figure}[t]
\centering
\includegraphics[width=\textwidth]{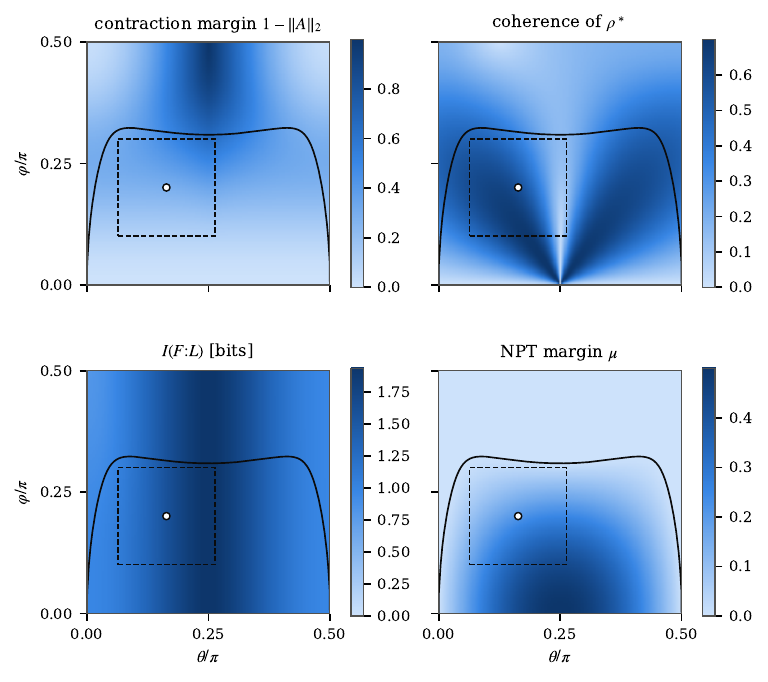}
\caption{The four properties on the slice $\kappa=\kappa_0$,
$\beta=\beta_0$ of Eq.~(\ref{eq:ref}). The black line is $Q=0$; the loop is
not EB below it. The white dot is the reference point, and the dashed
square is the region of Proposition~\ref{prop:square}.}
\label{fig:slice}
\end{figure}

Figure~\ref{fig:slice} shows the four properties on the slice through the
reference point. Contraction, information and coherence hold on the whole open
slice (the coherence approaches zero only towards the edge $\varphi=0$ and
towards one point, $\theta\approx0.127\pi$, of the edge
$\varphi=\pi/2$), so the binding condition is $Q>0$. The square of
Proposition~\ref{prop:square} is 77
times wider in each direction than the square certified in
\cite{Kopel2026a}; the largest square about the reference point on which
the loop stays non-EB has half-width $0.1086\pi$.

\section{Discussion}

Theorem~\ref{thm:record} is a design rule for memoryless loops: every
register that a round discards is a which-way detector for the message, and
the negativity of the loop is capped by half the fidelity of the records.
The responsiveness of the future to the message, which is what makes a loop
self-referential, is such a record. In the three-qubit loop the exact
boundary (\ref{eq:Q}) shows how the two records, the future event and the
probe after the swap, combine, and that they do not simply add up: a
partial record in the future can be read by the probe and returned by the
swap.

The results have clear limits. The exact solution concerns one
three-qubit model; only Theorem~\ref{thm:record} is general. The rounds are
memoryless, with fresh auxiliaries; a loop whose environment carries
memory between rounds is described by a process tensor with memory
\cite{MilzModi2021,Pollock2018}, and its rounds do not compose as powers of
one channel. Coherence is measured in a fixed basis, and a stronger
treatment would tie it to an operational task \cite{Streltsov2017}. A
non-EB loop need not have positive quantum capacity. Finally, the future is
internal to the simulated process; nothing here signals into the external
past. Open questions include tight bounds for $d$-level messages, whether
the re-entrance of Corollary~\ref{cor:reentrant} persists in larger loops,
and an operational meaning for the region where all four properties hold.

\section*{Appendix A. Derivation of the closed forms}

Write $s_\theta=\sin\theta$ and $c_\theta=\cos\theta$, so that
$|\psi_0\rangle=(s_\theta,c_\theta)$ and $|\psi_1\rangle=(c_\theta,
s_\theta)$, and $|\lambda_f\rangle=(\cos\frac\kappa2,(-1)^f\sin
\frac\kappa2)$. Without the controller ($\beta=0$; we write
$\tilde\Phi$ for this channel), the round maps $|m\rangle_M|0\rangle_F
|0\rangle_L$ to
\begin{equation}
|v_m\rangle=\sum_f\psi_m(f)\,\big[\cos\varphi\,|m\rangle_M|\lambda_f
\rangle_L-i\sin\varphi\,|\lambda_f\rangle_M|m\rangle_L\big]\otimes
|f\rangle_F.
\label{eq:A1}
\end{equation}
Tracing out $F$ and $L$,
\begin{align*}
\tilde\Phi(|m\rangle\langle n|)={}&\sum_f\psi_m(f)\psi_n(f)\Big[
\cos^2\varphi\,|m\rangle\langle n|
+\sin^2\varphi\,\langle n|m\rangle\,|\lambda_f\rangle\langle\lambda_f|\\
&+i\sin\varphi\cos\varphi\big(\langle n|\lambda_f\rangle\,
|m\rangle\langle\lambda_f|-\langle\lambda_f|m\rangle\,
|\lambda_f\rangle\langle n|\big)\Big].
\end{align*}
With $u=\sin^2(\kappa/2)$, $S=\sin^2\varphi$ and
$\eta=\frac12\sin\kappa\sin\varphi\cos2\theta$, this gives
\begin{equation}
\sigma_0=\begin{pmatrix}1-uS & -ie^{-i\varphi}\eta\\
ie^{i\varphi}\eta & uS\end{pmatrix},\qquad
\sigma_1=\begin{pmatrix}(1-u)S & -ie^{i\varphi}\eta\\
ie^{-i\varphi}\eta & 1-(1-u)S\end{pmatrix}
\label{eq:A2}
\end{equation}
for $\sigma_m=\tilde\Phi(|m\rangle\langle m|)$, and
\begin{equation}
\tilde\Phi(|1\rangle\langle0|)=\gamma\,|1\rangle\langle0|,
\label{eq:A3}
\end{equation}
with $\gamma$ as in (\ref{eq:gamma}). For (\ref{eq:A3}), note that
$\psi_1(f)\psi_0(f)=s_\theta c_\theta$ for both $f$, that the term with
$\langle n|m\rangle$ vanishes, and that $\sum_f\langle0|\lambda_f\rangle
\langle\lambda_f|=2\cos^2\frac\kappa2\,\langle0|$ and
$\sum_f\langle\lambda_f|1\rangle|\lambda_f\rangle=2\sin^2\frac\kappa2\,
|1\rangle$.

{\it Bloch data.} From (\ref{eq:A3}),
$\tilde\Phi(X)=\mathrm{Re}\gamma\,X+\mathrm{Im}\gamma\,Y$ and
$\tilde\Phi(Y)=-\mathrm{Im}\gamma\,X+\mathrm{Re}\gamma\,Y$; moreover
$\tilde\Phi(Z)=\sigma_0-\sigma_1$ and $\tilde\Phi(1/2)=(\sigma_0+
\sigma_1)/2$. Reading off the Pauli components gives
Eq.~(\ref{eq:closed}). The controller rotates every output Bloch vector by
$R_\beta$.

{\it Determinant.} In the basis $|n\rangle_R|j\rangle_M$, ordered
$00,01,10,11$, the partial transpose of the Choi state of $\tilde\Phi$ is
$\frac12T$, with
\begin{equation}
T=\begin{pmatrix}a_0 & b_0 & 0 & 0\\ \bar b_0 & d_0 & \gamma & 0\\
0 & \bar\gamma & a_1 & b_1\\ 0 & 0 & \bar b_1 & d_1\end{pmatrix},
\label{eq:A4}
\end{equation}
where $\sigma_0=\left(\begin{smallmatrix}a_0 & b_0\\ \bar b_0 & d_0
\end{smallmatrix}\right)$ and $\sigma_1=\left(\begin{smallmatrix}a_1 & b_1
\\ \bar b_1 & d_1\end{smallmatrix}\right)$. Expanding along the first row,
\begin{equation}
\det T=\det\sigma_0\det\sigma_1-|\gamma|^2a_0d_1.
\label{eq:A5}
\end{equation}
By (\ref{eq:A2}), $\det\sigma_0=uS\,[1-uS-(1-u)C]$,
$\det\sigma_1=(1-u)S\,[1-(1-u)S-uC]$, $a_0d_1=(1-uS)(1-(1-u)S)$, and by
(\ref{eq:gamma}), $|\gamma|^2=(1-C)(1-S)(1-4qS)$. Inserting these in
(\ref{eq:A5}) and using $\det J^\Gamma=\det T/16$ gives
Theorem~\ref{thm:boundary}. With the controller, $J^\Gamma$ is conjugated
by a unitary on the output factor, which changes neither the determinant
nor the spectrum.

{\sloppy All identities of this appendix are checked exactly by the
script \texttt{derive\_closed\_forms.py}, which builds the round from its
gates with $\sin$ and $\cos$ of the three angles as polynomial variables
and reduces every difference to zero modulo the Pythagorean
relations.\par}

\section*{Data and code availability}

{\sloppy No data were generated. The ancillary files of arXiv:2608.13764v2
contain \texttt{derive\_closed\_forms.py}, an exact symbolic check of
Theorems~\ref{thm:closed} and~\ref{thm:boundary} and of
Eqs.~(\ref{eq:b}), (\ref{eq:gamma}) and (\ref{eq:A2})--(\ref{eq:A5});
\texttt{check\_closed\_forms.py}, a numerical comparison of every closed
form with the explicit $8\times8$ unitary at random points, with tests of
Theorem~\ref{thm:record}, Corollary~\ref{cor:duality} and
Remark~\ref{rem:tight} on random channels; \texttt{phase\_numbers.py},
for the numbers in Corollary~\ref{cor:reentrant},
Proposition~\ref{prop:generic}, Eq.~(\ref{eq:ref}) and the largest non-EB
square of Section~6;
\texttt{certify\_square.py}, for Proposition~\ref{prop:square}; and
\texttt{make\_figures.py}. The interval certificate of the earlier version
is in the repository
\url{https://github.com/erankopel/quantum-future-feedback} (release v1.0.1
and later).\par}

\section*{Funding and competing interests}

No funding was received for this work. The author declares no competing
interests.

\section*{Declaration of AI Use}

During the preparation of this work the author used Microsoft 365 Copilot,
for the first version, to assist with literature discovery, early drafting,
algebraic organization, numerical code and figures, and Anthropic Claude
(Claude Fable 5.1 and Claude Opus 5.5, through the Claude desktop app,
September and October 2026). For the present version, Claude was used to
propose and derive the bound of Theorem~\ref{thm:record} and the closed
forms of Sections 4 and 5, to write the verification and figure scripts, to
search and check the literature, and to draft and revise the text. The
author reviewed, verified and edited the output and takes full
responsibility for the content of the publication.

\end{document}